\documentclass[reqno]{amsart}
\newcommand{\heute}{14. August 2009}
\usepackage{color}
\usepackage{amssymb,amsfonts,verbatim}
\newcommand{\bb}[1]{{\mathbb{#1}}}
\newcommand{\be}{\begin{equation}}
\newcommand{\ee}{\end{equation}}
\newcommand{\e}{\operatorname{e}}
\newcommand{\tp}{\tilde{\psi}}
\renewcommand{\Im}{\operatorname{Im}}
\renewcommand{\Re}{\operatorname{Re}}
\newtheorem{thm}{Theorem}[section]
\newtheorem{cor}[thm]{Corollary}
\newtheorem{lemma}[thm]{Lemma}
\theoremstyle{definition}
\newtheorem{dfn}{Definition}
\newtheorem{Remark}{Remark}
\begin{document}
\title[On the inverse resonance problem]{On the inverse resonance problem for Schr\"odinger operators}
\author{Marco Marlettta, Roman Shterenberg, Rudi Weikard}
\address{MM: School of Mathematics, Cardiff University, Cardiff CF24 4AG, Wales}
\email{MarlettaM@cardiff.ac.uk}
\address{RS: Department of Mathematics, University of Alabama at Birmingham, Birmingham, AL 35226-1170, USA}%
\email{shterenb@math.uab.edu}%
\address{RW: Department of Mathematics, University of Alabama at Birmingham, Birmingham, AL 35226-1170, USA}%
\email{rudi@math.uab.edu}%

\thanks{R.S. would also like to thank Erwin Schr\"odinger Institute (ESI, Wien) for its hospitality.}

\date{\heute}

\begin{abstract}
We consider Schr\"{o}dinger operators on $[0,\infty)$ with compactly supported, possibly complex-valued potentials in $L^1([0,\infty))$. It is known (at least in the case of a real-valued potential) that the location of eigenvalues and resonances determines the potential uniquely. From the physical point of view one expects that large resonances are increasingly insignificant for the reconstruction of the potential from the data. In this paper we prove the validity of this statement, i.e., we show conditional stability for finite data. As a by-product we also obtain a uniqueness result for the inverse resonance problem for complex-valued potentials.
\end{abstract}


\maketitle

\section{Introduction}\label{section:0}
Inverse scattering theory as well as inverse spectral theory for the Schr\"odinger equation
$$-y''+qy =\lambda y$$
are classical subjects, its central tenets having been established some 60 years ago by Borg, Levinson, Gelfand and Levitan, Krein, and Marchenko\footnote{See, for example, the monographs by Levitan \cite{MR933088}, Marchenko \cite{MR897106}, or Naimark \cite{MR0262880}.}. During this time a vast body of literature on the subject has been created. A particular class of problems, not quite so well-established, are the inverse resonance problems which are formulated only for a much narrower class of potentials. However, they are highly interesting from the point of view of applications since eigenvalues and resonances are directly observable in spectrometers.

We begin by expanding somewhat on the basics of scattering theory\footnote{For more details see Chadan and Sabatier \cite{MR985100} or Newton \cite{MR1947260}.}. To be specific, the equation $-y''+qy =\lambda y$ will be considered on the half-line $[0,\infty)$ with a Dirichlet boundary condition at zero, a case of considerable importance (thanks to separation of variables) for potential scattering in three dimensions with a spherically symmetric potential. When $q$ is integrable there is a unique solution of $-y''+qy=z^2y$ which behaves asymptotically like $\e^{izx}$ as long as $z$ is in the upper half plane. This solution is called the Jost solution and is denoted by $\psi(z,\cdot)$. At least when $q$ is superexponentially decaying in the mean, i.e., when $\int_0^\infty \e^{rx}|q(x)|dx$ is finite for all positive $r$, the function $\psi(\cdot,x)$ may be extended to the complex plane as an entire function for any fixed $x\in[0,\infty)$ (see, for example, Lemma 3 in \cite{MR1994689}). If, for some $C>0$ and $p>1$, the potential $q$ satisfies $|q(x)|\leq C \e^{-x^p}$, the growth order of $\psi(\cdot,x)$ is at most $p/(p-1)$. The function $\psi(\cdot,0)$ is called the Jost function and is of central importance in scattering theory. If $z_0$ is a zero of the Jost function in the upper half plane then, due to its asymptotic behavior, the corresponding Jost solution is an eigenfunction of the Schr\"odinger operator associated with the eigenvalue $z_0^2$. The zeros of the Jost function in the lower half plane are also of physical importance. They (or rather their squares) are called resonances. Our interest in eigenvalues and resonances stems from the fact that they are fundamental objects in quantum physics with a long history, dating back to the early days of the theory when Weisskopf and Wigner \cite{56.0752.06} studied the behavior of unstable particles. Physically, eigenvalues represent states in which the particles are permanently localized. Resonances, however, correspond to quasi-stationary (metastable) states that only exist for a finite time, proportional to the inverse of the imaginary part of the resonance, and have energy proportional to the real part of the resonance. Resonances that are close to the real axis appear as bumps in the scattering cross section and  can be measured in the laboratory. For more details on resonances the reader may consult Zworski \cite{MR2000d:58051}.

If the Jost function $\psi(\cdot,0)$ is of finite growth and the location of all eigenvalues and all resonances is known, then Hadamard's factorization theorem implies that it is known up to a factor $\e^{P(z)}$ where $P$ is a polynomial. But the coefficients of $P$ are determined also since it is known that the Jost function tends to one as $z$ tends to infinity along the positive imaginary axis for any potential under consideration. The Jost function, in turn, determines directly the norms of the Jost solutions associated with eigenvalues and a quantity called the scattering phase. Marchenko's inverse scattering theorem states that the eigenvalues, the norming constants, and the scattering phase determine uniquely the potential of the Schr\"odinger equation, assuming it is real-valued and satisfies the moment condition $\int_0^\infty x|q(x)|dx<\infty$. Thus, as a corollary, we have that the location of all eigenvalues and resonances determine a real superexponentially decaying potential. To our best knowledge this observation was first publicly made by Korotyaev in 2000 (published in \cite{MR2047740}) but Zworski \cite{MR1856251} had realized (but not published) it earlier in the context of compactly supported even potentials on $\bb R$.

We shall call the problem of obtaining a potential from just the location of all eigenvalues and resonances an inverse resonance problem. Inverse resonance problems are eminently interesting in a practical sense since, as mentioned before, eigenvalues and (small) resonances are attainable in the laboratory. This is in contrast to the scattering phase which is not easily measured. But, admittedly, finding all resonances --- as the theorem requires --- is just as elusive a goal as finding the scattering phase. Therefore a fundamental question arises: \emph{What may be said about a potential when only the location of the associated eigenvalues and small resonances is known.}

The following argument shows that large resonances carry little physical information. If the potential $q$ is compactly supported in $[0,b]$, absolutely continuous on its support, and if $q(b)\neq0$ then it is well known that there are only finitely many eigenvalues and that the (roots of the) resonances $x_n+iy_n$ are asymptotically close to the curve given by
$$y=\frac{-1}{2b}\log\left(\frac{4x^2}{|q(b)|}\right).$$
This shows that the asymptotic distribution of resonances changes upon the minutest change of the potential near the right endpoint of its support and suggests that one might be able to say a good deal about the potential from knowing the location of eigenvalues and resonances of modest size and, in particular, without knowledge of the asymptotic distribution of the large ones.

Indeed, despite the  fact that the finite data inverse resonance problem (or the finite data eigenvalue problem, on compact intervals) is ill-posed, having no unique solution, many numerical methods have been developed for its solution. Without a claim to completeness we mention here Andrew \cite{MR2277530}, Brown, Samko, Knowles, and Marletta \cite{MR1964261}, Hald \cite{MR505878}, Paine \cite{MR85c:65097}, R\"ohrl \cite{MR2183664}, Rundell and Sacks \cite{MR92h:34034}. The question now is how close these solutions are to each other and, more importantly, how close they are to the actual potential. The usual way to answer this question in the numerical analysis literature is to apply the recovery algorithm to a situation in which finite spectral data were generated from some known potential. The quality of the recovery procedure is then assessed according to how closely the recovered potential approximates the original one in some norm (sometimes the ``eyeball'' norm). Mathematically this involves are large leap of faith (practically such a leap of faith may, of course, be necessary).

Surprisingly it seems that this stability problem has not received much attention (one exception is Hitrik \cite{MR1759798}). Even for the much simpler inverse eigenvalue problem on a compact interval, this question was only answered as recently as \cite{MR2158108}. In \cite{MR2348728} it was addressed for a discrete Schr\"odinger equation. Some stability results for the case of a real-valued potential are given by Korotyaev \cite{MR2104289} but these do not address the case of finite resonance data, and indeed require quite delicate knowledge of the large-resonance asymptotics which will certainly not be available if only finitely many resonances are known.

In this paper we allow the potentials to be complex-valued but assume that they lie in a ball of fixed radius in $L^1(0,1)$ and are compactly supported. We suppose that two potentials $q$ and $\tilde q$ are both known to have compact support in some fixed interval - without loss of generality, the interval $[0,1]$ - and that for some $R>0$ and for some $\epsilon>0$, their resonances and eigenvalues are $\epsilon$-close inside the disc of radius $R$ centered at zero. Outside the disc of radius $R$, the resonances and eigenvalues of $q$ and $\tilde q$ need not be close at all; no assumption on the resonances outside the disc is made. If we also assume that $q-\tilde q$ is in a ball of fixed radius in $L^p(0,1)$, $p>1$, we obtain an estimate on
\begin{equation} \label{090415.1}
\sup_{x}\left|\int_x^1 (q(t)-\tilde q(t))dt \right|
\end{equation}
in terms of $\epsilon$ and $R$ which tends to zero as $1/R$ and $\epsilon$ tend to zero. Even without the assumption that $q-\tilde q\in L^p(0,1)$ we can show that (\ref{090415.1}) tends to zero as $1/R$ and $\epsilon$ tend to zero but we have no control over the rate of convergence. Note that in the case where all eigenvalues and resonances for $q$ coincide with those for $\tilde q$ we obtain a uniqueness result, i.e., $q$ and $\tilde q$ are then equal almost everywhere. To the best of our knowledge this uniqueness result is new for complex-valued potentials.

Our bound is obtained using transformation operators, Hadamard factorization, estimates developed from Jensen's formula and some elementary facts about Fourier transforms. The bulk of the work, however, lies in the estimation of the transformation operators from the resonance data. This is based on solving a non-standard boundary value problem for the hyperbolic PDE satisfied by the transformation kernels, with data given on a characteristic line. Our main result is Theorem \ref{theorem:4} together with its corollaries \ref{C:6.2} and \ref{C:6.3}.

\section{Transformation operators}\label{section:1}
\begin{dfn}
Let $Q>0$. By $B(Q)$ we denote the set of all (complex valued) functions $q\in L^1[0,\infty)$ which
have compact support in $[0,1]$ and are such that $\int_0^1 |q| \leq Q$.
\end{dfn}
Throughout this work we shall consider Schr\"{o}dinger operators on the half-line $[0,\infty)$ in which
the potentials lie in $B(Q)$. Given any $q\in B(Q)$ and any $z\in \bb C$ we consider the Schr\"{o}dinger equation
\[ -u'' + q(x) u = z^2 u, \;\;\; x > 0. \]
Since $q$ is compactly supported in $[0,1]$, for each $z\in\bb C$ this equation has a unique solution
$u$ satisfying the condition
\[ u(z,x) = \exp(izx), \;\;\; x\geq 1. \]
This solution is called the Jost solution and, for each fixed $x\geq 0$, it is an entire function of
$z$ of growth order one. We shall denote it by $\psi(z,x)$. Note that if $\Im(z)>0$ then
$\psi(z,\cdot)\in L^2[0,\infty)$, and
so if $\psi(z,\cdot)$ happens to satisfy the Dirichlet boundary condition $\psi(z,0) = 0$ then
$z^2$ will be an eigenvalue of the Dirichlet Schr\"{o}dinger operator $H_0(q)$ defined by
\[ H_0(q)u = -u'' + q u \]
on the domain
\[ {\mathcal D}(H_0(q)) = \{ u \in L^2[0,\infty) \; | \: -u'' + q u \in L^2[0,\infty), \;\; u(0) = 0\}. \]
If, on the other hand, $\psi(z,0)=0$ and  $\Im(z)\leq0$, then $\psi(z,\cdot)\not\in L^2[0,\infty)$ and
so $\psi(z,\cdot)$ cannot be an eigenfunction of $H_0(q)$. In this case $z^2$ is called a
{\em resonance} of $H_0(q)$. Thus we have the following dichotomy
of the zeros of $z\mapsto \psi(z,0)$:
\begin{itemize}
\item If $\Im(z)>0$ and $\psi(z,0)= 0$ then $z^2$ is an eigenvalue of $H_0(q)$ with eigenfunction
 $\psi(z,\cdot)$.
\item If $\Im(z)\leq 0$ and $\psi(z,0) = 0$ then $z^2$ is a resonance of $H_0(q)$ with wave function
 $\psi(z,\cdot)$.
\end{itemize}
Where no confusion will result, we shall abbreviate $\psi(z,0)$ simply to $\psi(z)$ and call this the
Jost function.

A remarkable fact about Schr\"{o}dinger equations -- see, e.g., Levitan \cite{MR933088} -- is the
existence of an integral operator $K$, not depending on $z$, which maps Jost solutions for
one potential to Jost solutions for a different potential. Given two potentials $q_1$ and $q_2$,
let $\psi_j(z,x)$ be the Jost solution of the equation  $-u'' + q_j u = z^2 u$, $j=1,2$. Then
there exists a kernel $K$ such that
\[ \psi_2(z,x) = \psi_1(z,x) + \int_x^\infty K(x,t)\psi_1(z,t)dt. \]
Throughout the later sections of this article we shall require estimates of kernels such as $K$. Since, in our situation,
$$\int_0^{(t-x)/2} |q_2(\alpha-\beta)-q_1(\alpha+\beta)|d\beta\leq 2Q$$
one obtains (see, e.g., Theorem 3 and Lemma 1 of \cite{MR2065435}) that
\be K(x,t)=\sum_{n=0}^\infty K_n(x,t) =: K_0(x,t) + H(x,t)\label{eq:Hdef}\ee
where
$$K_0(x,t)=\frac12 \int_{(t+x)/2}^1 (q_2(s)-q_1(s))ds$$
and, for $n\in\bb N$,
\[ K_n(x,t)=\int_{(t+x)/2}^1 \int_0^{(t-x)/2} (q_2(\alpha-\beta)-q_1(\alpha+\beta))
 K_{n-1}(\alpha-\beta,\alpha+\beta) d\beta d\alpha \]
so that
$$|K_n(x,t)| \leq \frac12 \frac{(2Q)^n}{n!} \left[1-\frac{t+x}2\right]_+^n \int_{(t+x)/2}^1 |q_2(s)-q_1(s)|ds$$
and
\[ |K(x,t)|  \leq  {\displaystyle \frac12 \int_{(t+x)/2}^1 |q_2(s)-q_1(s)|ds \exp(2Q\left[1-(t+x)/2\right]_+)}.
\]
In particular,
\begin{equation} \label{090123.1}
|K(x,t)|  \leq Q\e^{2Q}
\end{equation}
and $K_n(x,t)=K(x,t)=0$ if $t+x\geq2$. Notice also that $K_n(x,x)=0$ for $n\geq 1$ and hence
\be K(x,x) = K_0(x,x) = \frac{1}{2}\int_x^1 (q_2(s)-q_1(s))ds, \label{090123.2}
\ee
which allows the difference of the potentials to be recovered from the transformation kernel $K$.

It is also shown in \cite{MR2065435} that the function $t\mapsto H_t(0,t)$ is absolutely continuous, and
\be | H_t(0,t) | \leq CQ\e^{2Q}, \label{eq:Htbound}\ee
for some constant $C$ independent of $q_1$ and $q_2$.

As explained in the introduction, we consider in this paper the problem of estimating the difference
between two potentials $q$ and $\tilde{q}$ whose resonances are close to each other, if they are
not far from the origin. In order to do this we adopt some notation for specific transformation operator kernels corresponding to different choices of $q_1$ and $q_2$ above:

\begin{itemize}
\item $K_q$ for the transformation from a potential $0$ to a potential $q$.
\item $L_q$ for the transformation from a potential $q$ to the potential $0$.
\item $K_{\tilde{q}}$ for the transformation from a potential $0$ to a potential $\tilde{q}$.
\item $B$ for the transformation from a potential $q$ to a potential $\tilde{q}$.
\end{itemize}
We denote the Jost solution for the potential $q$ by $\psi$ and for $\tilde{q}$ by $\tp$.
Thus we have
\be \psi(z,x) = \exp(izx) + \int_x^{2-x} K_{q}(x,t)\exp(izt)dt \label{eq:2} \ee
and
\be \tp(z,x) = \exp(izx) + \int_x^{2-x} K_{\tilde{q}}(x,t)\exp(izt)dt. \label{eq:t2} \ee
Correspondingly, the kernel $L_q$ maps Jost solutions for potential $q$ back to
solutions of the free problem with potential 0:
\be \exp(izx) = \psi(z,x) + \int_x^{2-x} L_{q}(x,t)\psi(z,t)dt. \label{eq:3} \ee
Into the right hand side of (\ref{eq:t2}) we insert the expression for $\exp(izx)$
from (\ref{eq:3}) to obtain
\begin{eqnarray*} \tp(z,x) & = & \psi(z,x) + \int_x^{2-x} L_q(x,t)\psi(z,t)dt + \int_x^{2-x} K_{\tilde{q}}(x,t)\psi(z,t)dt \\
 & & \\
 & + & \int_x^{2-x} ds \int_s^{2-s} K_{\tilde{q}}(x,s)L_q(s,t)\psi(z,t)dt. \end{eqnarray*}
This means that
\[ \tp(z,x) = \psi(z,x) + \int_x^{2-x} B(x,t)\psi(z,t)dt, \]
in which
\[ B(x,t) = K_{\tilde{q}}(x,t)+L_q(x,t) + \int_x^t K_{\tilde{q}}(x,s)L_q(s,t)ds. \]
This expression is standard and may be found in \cite{MR933088}.
In the special case in which $\tilde{q} = q$ we know that $B$ must be zero and $K_q$ must
be $K_{\tilde{q}}$; this yields
\[ 0 = K_q(x,t)+L_q(x,t) + \int_x^t K_q(x,s)L_q(s,t)ds.\]
In particular, this gives
\begin{equation} \label{090328.2}
B(0,t) =  K_{\tilde{q}}(0,t)-K_q(0,t) + \int_0^t (K_{\tilde{q}}(0,s)-K_q(0,s))L_q(s,t)ds.
\end{equation}
We know from (\ref{090123.1}) that the sup-norm of transformation kernel $L_q$ is bounded by a constant which depends only on $\|q\|_1$. Thus we obtain a bound on $B(0,t)$ from one on $(K_{\tilde{q}}-K_q)(0,\cdot)$. This estimate on $B(0,t)$, in turn, will eventually yield a bound for $B(x,t)$ by an iterative procedure; the difference $\tilde{q}-q$ is then found from (\ref{090123.2}), which yields
\be 2B(x,x) = \int_x^1 (\tilde{q}-q). \label{eq:12} \ee

In order to find a bound on $(K_{\tilde{q}}-K_q)(0,\cdot)$ we observe that, setting $x=0$ in
eqns. (\ref{eq:2}), (\ref{eq:t2}) and inverting the Fourier transform,
\be (K_{\tilde{q}}-K_q)(0,t) = \frac{1}{2\pi}\int_{\bb R}(\tp(z)-\psi(z))\exp(-izt)dz. \label{eq:finvert} \ee
Our first task, therefore, is to estimate $\tp(z)-\psi(z)$ for real $z$. We shall start by doing this
in the case when $\psi$ and $\tp$ have exactly the same zeros in some large disc; the resulting
bound on $|(K_{\tilde{q}}-K_q)(0,\cdot)|$ is in Theorem \ref{theorem:m1}. The case where
the zeros inside the disc are perturbed is handled in Section \ref{section:8}.

\section{Estimates on the difference of Jost functions having the same zeros in a disc of radius $R$}\label{section:2}
In this section we will derive a pointwise bound for the difference of two Jost functions in the interval $[-R^{1/3},R^{1/3}]$ under the assumption that eigenvalues and resonances in the disc of radius $R$ coincide. More precisely, we will prove the following theorem.

\begin{thm} \label{T3.1}
For any positive number $Q$ there are numbers $C>0$ and $R_0\geq\e$ so that the following statement is true for any $R\geq R_0$. If $q$ and $\tilde q$ are two potentials in $B(Q)$ and if the zeros of the associated Jost functions $\psi$ and $\tilde\psi$ coincide in the disc $|z|<R$ then
$$|\psi(z)-\tilde\psi(z)|\leq C R^{-1/3}$$
for all $z$ satisfying $-R^{1/3}\leq z\leq R^{1/3}$.
\end{thm}

This theorem will be proved at the end of the section after several lemmas have been established. The key is Hadamard's factorization theorem which says that
$$\psi(z)=z^{n_0} \e^{g(z)}\prod_{n=1}^\infty E(z/z_n)$$
where $n_0$ is a nonnegative integer, $g$ is a polynomial of degree at most one, $E(w)=(1-w)\e^w$, and the $z_n$ are nonzero complex numbers. We introduce the abbreviation
$$\Pi(R,z)=\prod_{|z_n|\geq R}^\infty E(z/z_n).$$

We begin by establishing a preliminary estimate for the Jost function.
\begin{lemma} \label{L3.1}
For every positive number $Q$ there is a positive constant $\kappa$ such that the Jost function associated with any potential $q\in B(Q)$ has the following properties.
\begin{enumerate}
  \item $|\psi(z)|\leq\kappa$ for all $z\in\bb R$.
  \item $|\psi(z)|\leq \kappa \e^{2|z|}$ for all $z\in\bb C$.
  \item If $\rho>0$ then $|\psi(z)-1|\leq\kappa/\rho$ for all $z$ in the disc $\{z:|z-3i\rho|\leq\rho\}$.
\end{enumerate}
\end{lemma}

\begin{proof}
The representation $\psi(z)=1+\int_0^2 K_q(0,t)\e^{izt} dt$ gives immediately that
$|\psi(z)-1|\leq 2\|K_q(0,\cdot)\|_\infty \exp(2|\Im(z)|)$ for all $z\in\bb C$. This proves the first two statements. If $\Im(z)>0$ we may estimate $|\psi(z)-1|$ by
$\|K_q(0,\cdot)\|_\infty\int_0^2\e^{-t\Im(z)}dt \leq
\|K_q(0,\cdot)\|_\infty/\Im(z)$. Since the disc $\{z:|z-3i\rho|\leq\rho\}$ is contained in a sector where $\Im(z)\geq |z|/2$ we obtain the third statement.
\end{proof}

We now assume that $\rho\geq2\kappa$ and introduce the function $N(r)$ which counts the number of zeros of $\psi$ contained in the disc $\{z:|z-3i\rho|<r\}$. Note first that $N(0)=0$ since $|\psi(3i\rho)|\geq1/2$. Since $\psi$ has growth order one the counting function can only grow linearly. In fact, Jensen's formula
$$\int_0^{\e r} \frac{N(t)}tdt=\frac1{2\pi} \int_0^{2\pi} \log|\psi(3i\rho+\e r\e^{it})|dt-\log|\psi(3i\rho)|,$$
the inequality $N(r)\leq \int_0^{\e r}t^{-1} N(t) dt$, and part (2) of Lemma \ref{L3.1} give
\begin{equation} \label{090131.1}
N(r)\leq \log(2\kappa)+6\rho+2\e r.
\end{equation}

The elementary factor $E(w)=(1-w)\e^{w}$ satisfies  $|\log E(w)|\leq 2|w|^2$ as long as $|w|\leq 1/2$. Therefore, thinking of $w$ as $z/z_n$, we are interested in an estimate on
$$S=\sum_{|z_n|\geq R} |z_n|^{-2}.$$
It will be convenient to assume that $R\geq 9\rho\geq 18\kappa$. Since $|z_n|\geq R$ we get $|z_n|\geq 3|z_n-3i\rho|/4$ so that
$$S \leq 2 \sum_{|z_n-3i\rho |\geq 2R/3} |z_n-3i\rho|^{-2} \leq 2 \int_{2R/3}^\infty \frac{dN(t)}{t^2} \leq4\int_{2R/3}^\infty\frac{N(t)dt}{t^{3}}.$$
Using now inequality \eqref{090131.1} and $\log(2\kappa)\leq \rho/2$ gives $S\leq36/R$. With the aid of the inequality $|\e^u-1|\leq |u|\e^{|u|}$ we arrive at the following lemma.

\begin{lemma} \label{L3.2}
Let $z_n$, $n\in\bb N$, denote the nonzero zeros of the Jost function $\psi$ and assume that $R$ is a positive number which exceeds $18\kappa$, where $\kappa$ is the quantity given in Lemma \ref{L3.1}. Then
$$\left|\Pi(R,z)-1\right|\leq \frac{72|z|^2}{R} \exp(72|z|^2/R).$$
provided that $|z|\leq R/2$.
\end{lemma}

Now we return to the case of two potentials $q$ and $\tilde q$ in $B(Q)$. Let $\kappa$ be the number associated to $Q$ according to Lemma \ref{L3.1}. Since we assume the zeros of $\psi$ and $\tilde\psi$ within the disc of radius $R$ to coincide we get
$$\frac{\psi(z)}{\tilde\psi(z)}=\e^{g(z)-\tilde g(z)}\frac{\Pi(R,z)}{\tilde \Pi(R,z)}$$
and we need to estimate $\exp(g-\tilde g)$.
\begin{lemma} \label{L3.3}
There are positive constants $R_0$ and $c$ depending only on $\kappa$ such that
$$|\e^{g(z)-\tilde g(z)}-1|\leq  c R^{-1/3}$$
provided that $R\geq R_0$ and $|z|\leq R^{1/3}$.
\end{lemma}

\begin{proof}
Suppose $|z-3i\rho|\leq\rho$. Since
$$\e^{g(z)-\tilde g(z)}-1=\frac{\tilde \Pi(R,z)}{\Pi(R,z)} \left(\frac{\psi(z)}{\tilde\psi(z)}-1\right)+\frac{\tilde \Pi(R,z)-\Pi(R,z)}{\Pi(R,z)}$$
we get, when $\rho^2/R$ is sufficiently small,
$$|\e^{g(z)-\tilde g(z)}-1| \leq \frac{A\kappa}{\rho}+\frac{B\rho^2}{R}$$
from (3) of Lemma \ref{L3.1} and Lemma \ref{L3.2} when $A$ and $B$ denote suitable numerical constants. The two contributions to the error are in balance when we choose $\rho$ on the order of $R^{-1/3}$. Specifically, there are positive constants $R_0$ and $c$ depending only on $\kappa$ such that
$$|\e^{g(z)-\tilde g(z)}-1| \leq \frac{c}{10} R^{-1/3}$$
if $\rho^3=R\geq R_0$ and $|z-3i\rho|\leq\rho$.

Suppose $f(z)=\exp(a_1z+a_0)-1$. It is easy to show that
$$|f(z)|\leq \frac{5\varepsilon}{1-\varepsilon}\exp(5\varepsilon/(1-\varepsilon))$$
in the disc $|z|\leq\rho$ if $|f(z)|\leq\varepsilon<1$ in the disc $|z-3i\rho|\leq\rho$. Applying this to the case at hand gives the stated estimate after possibly increasing $R_0$ to ensure that $cR^{-1/3}\leq1$.
\end{proof}

\begin{proof}[Proof of Theorem \ref{T3.1}]
Suppose $-R^{1/3}\leq z\leq R^{1/3}$. By part (1) of Lemma \ref{L3.1}
$$|\psi(z)-\tilde\psi(z)|\leq \kappa \left|\frac{\psi(z)}{\tilde\psi(z)}-1\right| \leq \kappa |\e^{g(z)-\tilde g(z)}-1| \left|\frac{\Pi(R,z)}{\tilde\Pi(R,z)}\right| + \kappa \left|\frac{\Pi(R,z)}{\tilde\Pi(R,z)}-1\right|.$$
Using the estimates obtained in Lemma \ref{L3.2} and Lemma \ref{L3.3} establishes the theorem for a $C$ depending only on $c$, $R_0$, $\kappa$, and numerical constants and hence only on $Q$.
\end{proof}

\section{Large $z$ asymptotics of the Jost functions: further results and consequences for transformation kernels}\label{section:6}
In this section we assume that $q$ and $\tilde q$ lie in $B(Q)$ and that $\tilde q-q$ lies in $L^p[0,1]$ for some $p>1$. If $p>2$ then $\tilde q-q$ is still in $L^2([0,1])$ and therefore we assume henceforth that $p\in(1,2]$.

We start with (\ref{eq:finvert}) from Section \ref{section:1}:
\[ K_{\tilde{q}}(0,t)-K_q(0,t) = \frac{1}{2\pi}\int_{\bb R}(\tp(z)-\psi(z))\exp(-izt)dz. \]
In particular, therefore,
\be K_{\tilde{q}}(0,t)-K_q(0,t) = \frac{1}{2\pi}\int_{-R^{1/6}}^{R^{1/6}}\hspace{-12pt}(\tilde{\psi}-\psi)(z)\exp(-izt)dz
  + \frac{1}{2\pi}\int_{|z|>R^{1/6}}\hspace{-7mm}(\tilde{\psi}-\psi)(z)\exp(-izt)dz. \label{eq:6.2} \ee
The first term on the right hand side of (\ref{eq:6.2}) will be handled using Theorem \ref{T3.1},
which yields
\be \left| \frac{1}{2\pi}\int_{-R^{1/6}}^{R^{1/6}}(\tilde{\psi}-\psi)(z)\exp(-izt)dz \right| \leq
\frac{C}{R^{1/6}}. \label{eq:6.3}\ee
The second term,
\be E_{R}(t) := \frac{1}{2\pi}\int_{|z|>R^{1/6}}(\tilde{\psi}-\psi)(z)\exp(-izt)dz, \label{eq:6.4} \ee
will be handled using asymptotics which refine the results in Lemma \ref{L3.1} and
which we develop using the transformation equation (\ref{eq:2})
(with $x=0$) and integration by parts. Following the notation in (\ref{eq:Hdef}), we obtain, after an integration by parts,
$$\psi(z) = 1 + \frac{iK_q(0,0)}{z}-\frac{i}{4z}\int_0^2 g(t)\exp(izt)dt,$$
where $g(t)=q(t/2)-4H_t(0,t)$. We can also write
\[\psi(z) = 1 + \frac{iK_q(0,0)}{z}-\frac{i}{4z}\hat{g}(z) \]
where $\hat g$ is the Fourier transform of $g$. This immediately yields, in an obvious notation,
\begin{equation} \label{090328.1}
\tp(z)-\psi(z) = \frac{i}{z}(K_{\tilde{q}}-K_q)(0,0)-\frac{i}{4z}\widehat{(\tilde{g}-g)}(z).
\end{equation}
Since $H_t(0,t)$ is continuous and bounded by the bound given in (\ref{eq:Htbound}) $\tilde{g}-g$ is in $L^p([0,2])$ where $p\in(1,2]$ so that
$\widehat{(\tilde{g}-g)}$ is in $L^{p/(p-1)}(\bb R)$ by the Hausdorff-Young inequality.

Now we substitute (\ref{090328.1}) into the right hand side of (\ref{eq:6.4}) to obtain
\begin{multline}
E_{R}(t) = \frac{i}{2\pi}(K_{\tilde{q}}-K_q)(0,0)\int_{|z|>R^{1/6}}\frac1z \exp(-izt) dz\\
 - \frac{i}{8\pi}\int_{|z|>R^{1/6}}\frac1z \widehat{(\tilde{g}-g)}(z) \exp(-izt)dz. \label{eq:6.12}
\end{multline}
The first integral in (\ref{eq:6.12}) can be rewritten by a change of variable $\xi = tz$ as
$ \int_{|\xi|>tR^{1/6}}\frac{1}{\xi}\exp(-i\xi)d\xi  $, and hence an integration by parts
yields
\[  -2i\frac{\cos(tR^{1/6})}{tR^{1/6}} + i\int_{|\xi|>tR^{1/6}}\exp(-i\xi) \frac{d\xi}{\xi^2}, \]
which is $O((tR^{1/6})^{-1})$ when $tR^{1/6}$ is large. When $tR^{1/6}$ is small
the integral can be estimated by taking the Cauchy principal value. In either
case, there exists a numerical constant $C_1$ such that
\be \left|\frac{i}{2\pi}\int_{|z|>R^{1/6}}\frac{1}{z}\exp(-izt)dz\right| \leq C_1\min\left(1,\frac{1}{tR^{1/6}}\right). \label{eq:6.14} \ee
The remaining integral in (\ref{eq:6.12}) is estimated by the inequalities of H\"{o}lder and Hausdorff-Young so that
\be \left|\int_{|z|>R^{1/6}}\frac1z \widehat{(\tilde{g}-g)}(z) \exp(-izt)dz\right|
 \leq C_2 (p-1)^{-1/p} R^{(1-p)/(6p)}(1+\|\tilde q-q\|_p) \label{eq:6.15} \ee
where the constant $C_2$ depends only on $Q$ by (\ref{eq:Htbound}). Note that this becomes unbounded as $p$ tends to one.

Combining the estimates (\ref{eq:6.3}), (\ref{eq:6.14}) and (\ref{eq:6.15}) yields:
\begin{thm}\label{theorem:m1}
Suppose that $q$ and $\tilde{q}$ lie in $B(Q)$ and that $\tilde q-q$ is in $L^p[0,1]$, $p\in(1,2]$; suppose that $R_0$ is as in Theorem \ref{T3.1} and that $R\geq R_0$. Then there exists a constant $C$ (possible larger than the one used in  Theorem \ref{T3.1}) depending only on $Q$ such that
\be |K_{\tilde{q}}(0,t)-K_q(0,t)| \leq C(p-1)^{-1/p}(1+\|q-\tilde{q}\|_p)\min\left(1,\frac{1}{tR^{\nu}}\right)\ee
where $\nu =(p-1)/(6p)$.
\end{thm}

The relationship (\ref{090328.2}) gives immediately the following corollary.
\begin{cor} \label{T:4.2}
Under the assumptions of Theorem \ref{theorem:m1} we have the estimate
$$|B(0,t)|\leq C' (p-1)^{-1/p}(1+\|q-\tilde{q}\|_p)\min\left(1,\frac{\log R}{tR^{\nu}}\right)$$
where $\nu =(p-1)/(6p)$ and $C'$ is a constant depending only on $Q$.
\end{cor}

\begin{Remark}\label{remark:1}
If we do not require a rate-of-convergence estimate then we can assume simply that $q$ and $\tilde{q}$ lie in $L^1[0,1]$. In order to see this, consider the term
\[ F_{R}(t):= \int_{|z|>R^{1/6}}\frac{\hat{q}(z)}{z}\exp(-izt)dz. \]
Using the definition of the Fourier transform and Fubini's theorem, we get
\[ F_{R}(t) = 2i \lim_{n\to\infty} \int_0^1 q(s)\left(\int_{R^{1/6}}^n \frac{\sin(z(s-t))}{z}dz\right)ds. \]
The inner integral is bounded as function of $n$, $R$, $s$, and $t$. Hence, applying the dominated convergence theorem twice shows that
$F_{R}(t)$ tends to zero as $R$ tends to infinity.
\end{Remark}

\section{Estimating the difference of two potentials from the difference of two Jost functions} \label{section:7}
Eqn. (\ref{eq:12}) will yield bounds on $\int_x^1(\tilde{q}-q)$ from a bound on $B(x,x)$ and hence, in
particular, from a bound on $B(x,t)$. Corollary \ref{T:4.2} gives a bound on $B(0,t)$. In order to determine a bound on $B(x,t)$ from the bound on $B(0,t)$ we first observe that, following the material in Section \ref{section:1},
\[ B(x,t) = 0 \;\;\; \mbox{for $x+t\geq 2$}. \]
In particular,
\be B(x,2-x) = 0, \;\;\; 0\leq x \leq 1. \label{eq:18} \ee
We shall show that this condition, combined with the knowledge of $B(0,t)$ for $0\leq t\leq 2$,
determines $B$ completely. In order to do this we derive a second integral equation for $B$.

Following Levitan \cite{MR933088} we observe that the function
\[ z_0(x,t) = \frac{1}{2}\int_0^x du \int_{t+u-x}^{t-u+x}F(u,v)dv \]
satisfies the inhomogeneous wave equation
\[ z_{0,xx}-z_{0,tt} = F(x,t), \]
with boundary condition $z_0(0,t) = 0$. Putting $g(x) = z_0(x,2-x)$ we observe that
\[ v_0(x,t) = g\left(1-\frac{t+x}{2}\right)-g\left(1-\frac{t-x}{2}\right) \]
satisfies the homogeneous wave equation together with the boundary conditions $v_0(0,t)=0$,
$v_0(x,2-x) = -g(x) = -z_0(x,2-x)$. Hence the function $u_0(x,t) = z_0(x,t)+v_0(x,t)$ will satisfy
the inhomogeneous wave equation
\[ u_{0,xx}-u_{0,tt} = F(x,t) \]
with homogeneous boundary conditions
\[ u_0(0,t) = 0, \;\;\; u_0(x,2-x) = 0. \]
Now $u_0(x,t)$ is expressed as a sum of three integrals:
\[ u_0(x,t) = I_1(x,t)+I_2(x,t)-I_3(x,t), \]
where
\[ I_1(x,t) = z_0(x,t) = \frac{1}{2}\int_0^x du \int_{t+u-x}^{t-u+x}F(u,v)dv, \]
\[ I_2(x,t) = g\left(1-\frac{t+x}{2}\right) = \frac{1}{2}\int_0^{1-(t+x)/2}\hspace{-3mm}du
\int_{t+x+u}^{2-u}F(u,v)dv, \]
\[ I_3(x,t) = g\left(1-\frac{t-x}{2}\right) = \frac{1}{2}\int_0^{1-(t-x)/2}\hspace{-3mm}du
\int_{t-x+u}^{2-u}F(u,v)dv. \]
These integrals all have the same integrand and elementary calculations show that
the regions over which integrals $I_1$ and $I_2$ take place are disjoint and are sub-domains
of the region over which integral $I_3$ takes place. Subtracting $I_1$ and $I_2$ from $I_3$
leaves an integral over the rectangle with corners $(x,t)$, $(1-(t-x)/2,1+(t-x)/2)$,
$(0,t+x)$ and $(1-(t+x)/2,1+(t+x)/2)$. We make the change of variables
\[ v = \alpha + \beta, \;\;\; u = \alpha - \beta, \;\;\; du dv = 2 d\alpha d\beta , \]
and obtain
\[ u_0(x,t) = -\int_{(t+x)/2}^1 d\alpha \int_{(t-x)/2}^{(t+x)/2}F(\alpha-\beta,\alpha+\beta)d\beta.
\]
If we now ask for the solution $w$ of the wave equation
\[ w_{xx}-w_{tt} = F(x,t),\;\;\; w(0,t) = B(0,t), \;\;\; w(x,2-x)=0, \]
then, since $B(0,t)=0$ for $t\geq 2$,  $w$ will be given by
\be w(x,t) = u_0(x,t) + B(0,x+t). \label{eq:26}\ee

Now the transformation kernel $B$ is required to satisfy the wave equation
\[ (B_{x}-B_{t})_x+(B_x-B_t)_t=(\tilde{q}(x)-q(t))B \]
subject to the condition (\ref{eq:18}), and with $B(0,t)$ known. In view of the expression (\ref{eq:26})
this means we should choose $B$ as the solution of the integral equation
\[ B(x,t) = B(0,x+t) + \int_{(t+x)/2}^1 d\alpha \int_{(t-x)/2}^{(t+x)/2}(q(\alpha+\beta)-\tilde{q}(\alpha-\beta))B(\alpha-\beta,\alpha+\beta)d\beta.
\]
Iteration (cf. Lemma \ref{theorem:1} below) shows that this solution is given by
\[ B(x,t) = \sum_{n=0}^\infty B_n(x,t) \]
where
\[ B_0(x,t) = B(0,x+t) \]
and
\be B_{n+1}(x,t) =  \int_{(t+x)/2}^1 d\alpha \int_{(t-x)/2}^{(t+x)/2}(q(\alpha+\beta)-\tilde{q}(\alpha-\beta))B_n(\alpha-\beta,\alpha+\beta)d\beta.
\label{eq:29}\ee
\begin{lemma}\label{theorem:1}
Suppose that $q$ and $\tilde q$ are in $B(Q)$ and that there exist constants $C_0$, $C_1$ and $R_2\geq \e$ such that for all $t\in (0,2]$,
\be |B(0,t)| \leq C_1+C_0 \min\left(1,\frac{1}{tR_2}\right). \label{eq:30}\ee
Then
\be |B_n(x,t)| \leq \left(C_1+C_0\frac{\log(2R_2)}{R_2}\right)\frac{(2Q)^n}{(n-1)!}\left(1-\frac{t+x}{2}\right)^{n-1} \label{eq:31}\ee
whenever $n\in\mathbb{N}$ and $0\leq x+t \leq 2$.
\end{lemma}
\begin{proof} The proof is by induction. We first check that the estimate holds for $n=1$. We have
\begin{eqnarray*}
 |B_1(x,t)| & \leq & \int_{(t+x)/2}^1d\alpha \int_{(t-x)/2}^{(t+x)/2}|q(\alpha+\beta)-\tilde{q}(\alpha-\beta)||B(0,2\alpha)|d\beta \\
  & \leq & 2Q \int_{0}^1|B(0,2\alpha)|d\alpha  \\
  & \leq & 2QC_0 \left(C_1+C_0\frac{\log(2R_2)}{R_2}\right),
\end{eqnarray*}
which establishes the result for $n=1$.

Next we substitute the estimate (\ref{eq:31}) into the right hand side of (\ref{eq:29}) and try to recover the appropriate estimate for $B_{n+1}$. Since (\ref{eq:31}) holds we have
$$|B_{n+1}(x,t)| \leq \left(C_1+C_0\frac{\log(2R_2)}{R_2}\right)\frac{(2Q)^{n+1}}{(n-1)!} \int_{(t+x)/2}^1\hspace{-11pt}
 (1-\alpha)^{n-1}d\alpha$$
which yields the required estimate.
\end{proof}

\begin{lemma}\label{theorem:2}
Under the hypotheses of Lemma \ref{theorem:1} the estimate
\be |B(x,t)| \leq \left(C_1+C_0\frac{\log(R_2)}{(x+t)R_2}\right) (1+8Q\e^{2Q}) \label{eq:33}\ee
holds for all $(x,t)$ in the triangle bounded by the lines $x=0$, $x=t$ and $x+t=2$.
\end{lemma}
\begin{proof} This is an immediate consequence of the fact that $B(x,t) = B(0,x+t) + \sum_{n=1}^\infty B_n(x,t)$ together with the bounds (\ref{eq:30},\ref{eq:31}).
\end{proof}

\begin{thm}\label{theorem:3}
Let $Q_1$ and $Q_p$ be positive numbers and $p\in(1,2]$. Then there is a positive number $C$, depending only on $Q_1$ and $Q_p$, and a positive number $R_0$, depending on $Q_1$, $Q_p$, and $p$, so that the following statement is true for any $R\geq R_0$. If $q$ and $\tilde{q}$ are two potentials in $B(Q_1)$ such that $\|\tilde{q}-q\|_p\leq Q_p$ and for which the zeros of the corresponding Jost functions are identical in a disc of radius $R$, then
\be \sup_{x\in [0,1]}\left|\int_x^1 (\tilde{q}-q)\right| \leq C(Q_1,Q_p) (\log R)^{(2p-2)/(2p-1)}R^{-(p-1)^2/(6p(2p-1))}. \label{eq:34b} \ee
\end{thm}

\begin{proof}
Let $\gamma=(p-1)/p=6\nu$ and $0<\eta<1$. Lemma \ref{theorem:2}, with $C_1=0$ and $C_0$ determined by Corollary \ref{T:4.2}, implies the existence of a constant $C_2$, depending only on $Q_1$, such that
$$2|B(\eta,\eta)|\leq C_2 (p-1)^{-1/p}(1+Q_p)\frac{(\log R)^2}{\eta R^\nu} :=\frac{M}{\eta}$$
provided that $R$ is at least as large as $R_0$ given by Theorem \ref{T3.1}. By possibly enlarging $R_0$ we have $M<\gamma Q_p$ when $R\geq R_0$. (If $M$ is not much smaller than $\gamma Q_p$ our final estimate will be worse or not much better than the trivial estimate $\left|\int_x^1 (\tilde{q}-q)\right| \leq Q_p$).

Thus, if $0<\eta<1$ then, for all $x\in[0,1]$,
\begin{equation} \label{090328.3}
\left|\int_x^1 (\tilde{q}-q) \right| \leq \int_0^\eta |\tilde{q}-q| + \left|\int_\eta^1 (\tilde{q}-q)\right|
\leq Q_p \eta^{\gamma} +\frac{M}{\eta}
\end{equation}
according to equation (\ref{eq:12}). Substituting $\eta=(M/(\gamma Q_p))^{1/(\gamma+1)}<1$, the point where the best estimate occurs, into (\ref{090328.3}) gives the desired result. Note that $C(Q_1,Q_p)$ may be chosen, independently of $p$.
\end{proof}

\begin{Remark}
The proof of Theorem \ref{theorem:3} reveals that $C(Q_1,Q_p)=O(Q_p^{p/(2p-1)})$ as $Q_p$ tends to zero.
\end{Remark}

\section{Further errors from perturbation of the resonances} \label{section:8}
In the previous sections we considered a change of potential from $q$ to $\tilde{q}$ which preserved
all zeros of the Jost function in a disc of radius $R$. In the current section we allow the potentials
$q$ and $\tilde{q}$ to have different zeros $z_n$ and $\tilde{z}_n$ inside the disc of radius $R$,
satisfying a uniform bound:
\[ |z_n - \tilde{z}_n| \leq \epsilon , \;\;\; n = 1,2,\ldots N := N(R). \]
Since zeros of the Jost function in the upper half plane lie in a half disc centered at zero whose radius is bounded by a constant times $Q$
we may assume without loss of generality that none of the zeros $z_n$ or $\tilde{z}_n$ lie on the real
axis. If this is not true then the contours taken below for the inversions of the Fourier transforms
may be deformed around the zeros. We therefore assume that at every point $z$ on our inversion
contour, and for all $n$,
\[ |z-z_n| \geq 1, \;\;\; |z-\tilde{z}_n| \geq 1. \]
Define a function
\[ W(z) = \prod_{n=1}^{N(R)}\frac{(z-z_n)}{(z-\tilde{z}_n)}. \]
Then the Jost functions $\psi(z)$ and $\tilde{\psi}(z)$ satisfy
\[ \psi(z) = W(z)\tilde{\psi}(z)\frac{\mbox{e}^{g(z)}\Pi(R,z)}{\mbox{e}^{\tilde{g}(z)}\tilde{\Pi}(R,z)}
\]
in the notation of Section \ref{section:2}. We still have the equation
\[ K_q(0,t) - K_{\tilde{q}}(0,t) = \frac{1}{2\pi}\int_{\bb R}(\psi(z)-\tilde{\psi}(z))\exp(-izt)dz  =: I_{R}(t) + E_R(t),
\]
in which
\[ E_{R}(t) = \frac{1}{2\pi}\int_{|z|>R^{1/6}}(\psi(z)-\tp(z))\exp(-izt)dz, \]
\[ I_{R}(t) = \frac{1}{2\pi}\int_{-R^{1/6}}^{R^{1/6}}(\psi(z)-\tp(z))\exp(-izt)dz, \]
We also still have the estimate on $E_R$ from (\ref{eq:6.14}), (\ref{eq:6.15}) giving the existence of a constant $C_1$ depending only on $Q$ such that
\[ |E_{R}(t)| \leq C_1(1+\| q - \tilde{q}\|_p) \min\left(1,\frac{1}{t R^\nu}\right) \]
in which $\nu = (p-1)/(6p)$. We therefore turn to estimating $I_R(t)$, which we first write as
\[ I_R(t) = \frac{1}{2\pi}\int_{-R^{1/6}}^{R^{1/6}}(\psi(z)\left(1-
\frac{\mbox{e}^{\tilde{g}(z)}\tilde{\Pi}(R,z)}{\mbox{e}^{g(z)}\Pi(R,z)}\right)
 + \tilde{\psi}(z)(W(z)-1))\exp(-izt)dz. \]
Estimating the first term using Lemma \ref{L3.1} (part (1)) to bound $\psi(z)$, and
Lemmas \ref{L3.2}, \ref{L3.3} to bound $1-\frac{\mbox{e}^{\tilde{g}(z)}\tilde{\Pi}(R,z)}{\mbox{e}^{g(z)}\Pi(R,z)}$, we obtain
\[ |I_R(t)| \leq C_2 R^{-1/6} + \frac{1}{2\pi}\int_{-R^{1/6}}^{R^{1/6}}|\tp(z)||W(z)-1|dz \]
for some constant $C_2$ depending only on $Q$.
We now apply the inequality $|W-1|\leq |\log W|\exp(|\log W|)$ together with the bound $|\tp(z)|\leq
\kappa$ from Lemma \ref{L3.1} (part (1)) to obtain
\be |I_R(t)| \leq C_2 R^{-1/6} + \kappa\int_{-R^{1/6}}^{R^{1/6}}|\log W(z)|\exp(|\log W(z)|) dz. \label{eq:8.52}\ee
Since $|\log(1+x)|\leq -\log(1-|x|)\leq 2|x|$ whenever $|x|\leq 3/4$ we find
\begin{equation} \label{090330.1}
|\log W(z)|  \leq  \sum_{n=1}^{N(R)} \left|\log\left(1+\frac{\tilde{z}_n-z_n}{z-\tilde{z}_n}\right)\right|
\leq 2\sum_{n=1}^{N(R)} \left|\frac{\tilde{z}_n-z_n}{z-\tilde{z}_n}\right|
\end{equation}
provided the summands on the right are bounded by $3/4$. Denoting the smallest integer which is at least as large as $x$ by $\lceil x\rceil$ we distinguish now the cases $n\leq \lceil8{\rm e}R^{1/6}\rceil$ and $n>\lceil8{\rm e}R^{1/6}\rceil$. (Here we assume that $N(R)>\lceil8{\rm e}R^{1/6}\rceil$, since the case $N(R)\leq\lceil8{\rm e}R^{1/6}\rceil$ will give the same bounds a fortiori.) In the former case we use the assumption $|z-\tilde{z}_n| \geq 1$ to estimate a summand in (\ref{090330.1}) by $\epsilon$ which we assume to be less than $3/4$. In the latter case we observe that thanks to eqn. (\ref{090131.1}) with $\rho=2\kappa$, we have
\[ |z-\tilde{z}_n| \geq |\tilde z_n|-|z|\geq \frac{1}{2\mbox{e}}n-R^{1/6} - 12\kappa-\log(2\kappa)  \geq \frac{1}{4\e}n\]
if we require $R^{1/6}\geq 12\kappa+\log(2\kappa)$.
Hence, if $n>\lceil8{\rm e}R^{1/6}\rceil$ a summand in (\ref{090330.1}) is estimated by $4\e\epsilon/n$. Using (\ref{090131.1}) again we find
\[ |\log W(z)| \leq 16\e\epsilon R^{1/6} +8\e\epsilon \log(R)+\epsilon \leq 17\e \epsilon R^{1/6}. \]

Using this estimate for the argument of the exponential function in (\ref{eq:8.52}) we obtain
\be |I_R(t)| \leq C_2R^{-1/6} + \kappa\exp(17\e\epsilon R^{1/6})\int_{-R^{1/6}}^{R^{1/6}}|\log W(z)| dz. \label{eq:8.54}\ee
To estimate $\int |\log W|dz$ we stick to the bound in (\ref{090330.1}), i.e.,
\[ \int_{-R^{1/6}}^{R^{1/6}}|\log W(z)| dz \leq 2\epsilon(S_1+S_2), \]
where
\[ S_1 = \sum_{n=1}^{\lceil8{\rm e}R^{1/6}\rceil}\int_{-R^{1/6}}^{R^{1/6}}\frac{dz}{|z-\tilde z_n|}, \]
\[ S_2 = \sum_{n=\lceil8{\rm e}R^{1/6}\rceil+1}^{N(R)}\int_{-R^{1/6}}^{R^{1/6}}\frac{dz}{|z-\tilde z_n|}.
\]
In the sum $S_1$ we set $\xi=\Re(z-\tilde z_n)$ and observe that $|z-\tilde z_n|\geq\sqrt{1+\xi^2}$ with $|\xi|\leq2R$. Hence
\[ \int_{-R^{1/6}}^{R^{1/6}}\frac{dz}{|z-\tilde z_n|} \leq \int_{-2R}^{2R}\frac{d\xi}{\sqrt{1+\xi^2}} \]
giving
\be S_1 \leq c R^{1/6}\log(R) \label{eq:8.14}\ee
for some numerical constant $c$.
In the sum $S_2$ we use the same approach as before to obtain
\be S_2 \leq 16\e R^{1/6}\log(R). \label{eq:8.15}\ee
Combining (\ref{eq:8.54}), (\ref{eq:8.14}) and (\ref{eq:8.15}) we obtain
\[ |I_{R}(t)| \leq C_2R^{-1/6} + C_3\exp(17\e\epsilon R^{1/6})\epsilon R^{1/6}\log(R) \]
in which $C_2$ and $C_3$ depend only on $Q$. We thus obtain the total estimate:
\begin{multline*}
|K_{\tilde{q}}(0,t)-K_{q}(0,t)| \leq C\epsilon R^{1/6}\log(R) \exp(17 \e\epsilon R^{1/6}) \\ +C(p-1)^{-1/p}(1+\|\tilde{q}-{q}\|_p)\min\left(1,\frac{1}{tR^\nu}\right).
\end{multline*}
where $C$ depends only on $Q$. Upon possible enlarging $C$ equation (\ref{090328.2}) gives
\begin{multline*}
|B(0,t)| \leq C\epsilon R^{1/6}\log(R) \exp(17\e\epsilon R^{1/6}) \\ +C(p-1)^{-1/p}(1+\|\tilde{q}-{q}\|_p)\min\left(1,\frac{\log R}{tR^\nu}\right).
\end{multline*}

The final result follows from Lemma \ref{theorem:2} with $C_1= C\epsilon R^{1/6}\log(R) \exp(17\e\epsilon R^{1/6})$ and $C_0=C(p-1)^{-1/p}(1+\|\tilde{q}-{q}\|_p)\min\left(1,\frac{\log R}{tR^\nu}\right)$ using the same method proof as for Theorem \ref{theorem:3}:
\begin{thm}\label{theorem:4}
Let $Q_1$ and $Q_p$ be positive numbers and $p\in(1,2]$. Then there is a positive number $C$, depending only on $Q_1$ and $Q_p$, and a positive number $R_0$, depending on $Q_1$, $Q_p$, and $p$, so that the following statement is true for any $R\geq R_0$ and any $\epsilon\in(0,3/4)$. If $q$ and $\tilde{q}$ are two potentials in $B(Q_1)$ such that $\|\tilde{q}-q\|_p\leq Q_p$ and for which the respective zeros of the corresponding Jost functions are $\epsilon$-close in a disc of radius $R$, then
\begin{multline*}
\sup_{x\in [0,1]}\left|\int_x^1 (\tilde{q}-q)\right| \leq C(Q_1,Q_p) \Bigl\{(\log R)^{(2p-2)/(2p-1)}R^{-(p-1)^2/(6p(2p-1))}\\
+\epsilon R^{\frac{1}{6}}\log(R)\exp(17\e\epsilon R^{\frac{1}{6}})\Bigr\}.
\end{multline*}
\end{thm}

\begin{cor} \label{C:6.2} (Conditional stability). Let $q$ and $\tilde{q}$ be two potentials with support in $[0,1]$. Let $\int_0^1|q|^p<Q_p$ and $\int_0^1|\tilde{q}|^p<Q_p$ for some $p>1$. Then for any $\delta>0$ there exists
a pair $(\epsilon,R)$, depending only on $\delta$, $Q_p$, and $p$, such that if the corresponding Jost functions have zeros differing by at most $\epsilon$  in a disc of radius $R$ then
$$\sup\limits_{x\in[0,1]}\left|\int_x^1(\tilde{q}-q)\right|\leq\delta.$$
\end{cor}

Theorem \ref{theorem:4} and Remark \ref{remark:1}also imply the next corollary.
\begin{cor} \label{C:6.3} (Uniqueness)
Let $q$ and $\tilde{q}$ be two integrable potentials with support in $[0,1]$. If all eigenvalues and resonances for one potential coincide with those of the other, then $q$ equals $\tilde q$ almost everywhere.
\end{cor}

\begin{Remark} Taking into account Remark \ref{remark:1}, one can prove conditional stability for $q$, $\tilde{q}$ just from $B(Q)$. But in this case the radius $R$ can not be chosen uniformly. Indeed, consider the sequence of potentials $q_n(x):=n\chi_{[0,1/n]}(x)$ where $\chi$ is the characteristic function. Then for the corresponding Jost functions $\psi(z;n)$ we have
$$
\psi(z;n)=e^{iz/n}\left(\cos\left(\frac{\sqrt{z^2-n}}{n}\right)-i\frac{z}{\sqrt{z^2-n}}\sin\left(\frac{\sqrt{z^2-n}}{n}\right)\right).
$$
Obviously, for any $R$ there are no zeros of $\psi(z;n)$ in the disc $|z|<R$ if $n=n(R)$ is sufficiently large. Thus, one cannot choose $R$ to depend only on $\delta$ and $Q$. At the same time, in the same way as in the proof of the estimate (\ref{eq:34b}) one can still get uniformity for the class of potentials from $B(Q)$ if their behavior as $x\to0$ is specified. Namely, if we assume that $q,\,\tilde{q}\in B(Q)$ and
$$
\int_0^\epsilon(|\tilde{q}|+|q|)\leq \eta(\epsilon),
$$
where $\eta(\epsilon)\to0$ as $\epsilon\to0$, then $R$ can be chosen to depend only on $\delta$, $Q$, and the function $\eta$.
\end{Remark}

\subsection*{Acknowledgment:} The authors wish to thank Sergey Naboko for many discussions on the topic.


\end{document}